\documentclass{llncs}

\usepackage{mathrsfs}
\usepackage{amssymb,amsmath}
\usepackage{perpage} 
\usepackage{array}
\usepackage{cite}

\MakePerPage{footnote} 

\newtheorem{observation}{\bf Proposition}

\begin{document}
\title{\large \bf Asymptotic collusion-proofness of voting rules: the case of large number of candidates}
\author{Palash Dey \and Y. Narahari\\ \texttt{\{palash,hari\}@csa.iisc.ernet.in}}
\institute{Department of Computer Science and Automation \\Indian Institute of Science - Bangalore, India.}
\date{}
\maketitle
\thispagestyle{empty}

\begin{abstract}
Classical results in voting theory show that strategic 
manipulation by voters is inevitable if  
a voting rule simultaneously satisfy certain 
desirable properties.  Motivated by this, we study 
the relevant question of how 
often a voting rule is manipulable. 
It is well known that 
elections with a large number of voters are rarely manipulable
under impartial culture (IC) assumption. 
However, the manipulability of voting rules when the number
of candidates is large has hardly been addressed in the literature
and our paper focuses on this problem.
First, we propose two properties  (1) asymptotic strategy-proofness 
and (2) asymptotic collusion-proofness,
with respect to new voters, which makes
the two notions more relevant from the perspective 
of computational problem of manipulation. 
In addition to IC, we explore a new culture of society where all score vectors of the candidates
are equally likely. This new notion has its motivation in computational social choice and
we call it  impartial scores culture (ISC) assumption.
We study asymptotic strategy-proofness and asymptotic collusion-proofness for 
plurality, veto, $k$-approval, and Borda voting rules
under IC as well as ISC assumptions. Specifically,
we prove bounds for the fraction of manipulable profiles
when the number of candidates is large.
Our results show that the size of the coalition and the tie-breaking rule 
play a crucial role in determining whether or not a voting rule 
satisfies the above two properties.
\end{abstract}


\section{Introduction}
In many real life situations including multi-agent systems, agents often 
need to agree upon a common decision although they may have different 
preferences over the possible alternatives. The central problem in this 
scenario is preference aggregation, that is, aggregating different 
preferences and agreeing on a joint plan. Voting is a special type of 
preference aggregation method that has been applied 
in many current problems, for example, collaborative filtering \cite{pennock2000social},  
planning among multiple automated agents \cite{ephrati1991clarke}, etc.
A wide variety of voting rules have been proposed in the literature.
A voting rule in conjunction with the population of voters and candidates
is commonly referred to as an election.

A fundamental problem of voting rules is the strategic manipulation by voters 
- sometimes voters are better off by voting non-truthfully. 
Informally, a voter is said to \textit{manipulate} an election if she does not report 
her true preference. The Gibbard Satterthwaite theorem \cite{gibbard1973manipulation,satterthwaite1975strategy} 
shows that every \textit{unanimous}\footnote{A voting rule is called \textit{unanimous} if 
whenever any candidate is most preferred by all voters, that candidate is the 
winner.} and \textit{non-dictatorial}\footnote{A voting rule is called \textit{non-dictatorial} 
if there does not exist any voter whose most preferred candidate is always 
the winner irrespective of others' votes.} voting rule is manipulable. 
Clearly, manipulation is undesirable and therefore we seek voting 
rules that are \textit{rarely} manipulable. Pattanaik \cite{pattanaik1975strategic} 
conjectured that ``the possibility of strategic voting by single individuals will be 
smaller the greater the number of the individuals.''

Pattanaik's above conjecture has been subsequently verified. 
Slinko \cite{slinko2002asymptotic} showed that for some common voting 
rules that, under impartial culture (IC)\footnote{Impartial culture assumption says that 
the voters' preferences are independent and uniformly distributed among all possible 
linear orders of the candidates.} assumption, the probability of 
drawing a manipulable profile at random goes to zero as the number of voters increases. 
He called such voting rules \textit{asymptotically strategy-proof}. 
However the above result does \textit{not} directly connect the computational problem  
of manipulation since it defines manipulable profiles through 
the existence of manipulation by current voters in the profile. We, in this 
paper, propose a definition of asymptotic strategy-proofness with 
respect to new voters which makes it consistent with the 
\textit{computational problem of manipulation} which is discussed below.

Bartholdi and coauthors \cite{bartholdi1989computational,bartholdi1991single} 
asked the following question: what if finding a manipulating preference 
is computationally so hard that the agents will not be 
able to find it? This led to the formulation and study of the computational problem of 
\textit{coalitional manipulation} (CM). In the CM problem, we are 
given a voting rule, a candidate (say $x$), a set of $c$ manipulators, and a 
voting profile of $n$ non-manipulators on $m$ candidates, and we are asked whether the manipulators can make the 
candidate $x$ win the election or not. 
The CM problem for many common voting rules has been shown to be $\mathbb{NP}$-Complete 
\cite{faliszewski2008copeland,xia2010scheduling}. 
This computational intractability seems 
to provide a kind of barrier against manipulation for many common voting rules. 
However, since intractability only provides a \textit{worst case} guarantee, how much 
protection computational hardness provides in actual practice, is questionable. 
This has been partially answered by Procaccia 
and Rosenschein \cite{procaccia2007average} by showing average 
case easiness of the manipulation problem assuming \textit{junta} distributions over the  
voters' preferences. To the best of our knowledge, 
the following reverse question still remains unanswered - given the inputs where 
the CM problem is actually hard, how much severe is the problem of 
manipulation on those inputs? This question is justified as  
there is no point in trying to stop manipulation at profiles which are 
\textit{game theoretically} strategy-proof. 
We address this question in this paper. 

It is known that the CM problem can be solved in  $O\left(\left(c+m!\right)^{m!}.\text{\textit{poly}
}(m,n,c)\right)$ time\footnote{\textit{poly}$(m,n,c)$ denotes the set of all polynomial functions over $m,n$, and $c$}. for 
any anonymous and \textit{efficient}\footnote{A voting rule is \textit{efficient} 
if winner determination is \textit{poly(m,n)} time computable.} voting rule. The computational complexity 
arises from the fact that there are $O\left(\left(c+m!\right)^{m!}\right)$ different ways to distribute the $c$ votes 
of manipulators among $m!$ possible preferences. 
Hence all the hard instances of the CM problem are elections with a large number of candidates. 
This motivates us to study the severity of manipulation in elections with a large number of 
candidates. Asymptotic strategy-proofness has been classically studied 
keeping the number of candidates fixed and varying only the number 
of voters. Hence the classical asymptotic strategy-proofness \textit{fails} to connect 
itself with computational problem of manipulation since for constant number of 
candidates, complexity barrier does not exist. Also Nitzan \cite{nitzan1985vulnerability} 
empirically showed that the problem of manipulation is severe only in societies with 
small number of voters. Hence it is the elections with small number of voters 
which we should possibly target for preventing manipulation. However small elections with only a few candidates 
are always easily manipulable - manipulators can just try all possible linear orders 
over the candidates. This leaves only one case open - elections with fixed number of 
voters but large number of candidates. 

\subsection*{Contributions}
Besides active theoretical interest, there are many 
practical scenarios as well where the number of candidates could be large, for example, 
meta search engines, employee selection by a committee, etc. In this paper, we study 
the manipulability of voting rules when the number of voters is fixed and
the number of candidates increases to infinity. 
The specific contributions of this paper are as follows.
\begin{itemize}
\item We define the notions of asymptotic strategy-proofness 
 and asymptotic collusion-proofness with respect to new voters 
 (definition\nobreakspace \ref {newdefCollusionProofness}) which makes 
 the notions more relevant from the perspective 
 of computational problem of manipulation. We show that the existing results 
 on the manipulability for plurality, veto, and $k$-approval voting rules continue 
 to hold under the proposed definition of asymptotic strategy-proofness. 
 \item We explore a new culture of society where all score vectors of the candidates 
 are equally likely. This new notion has its motivation in computational social choice. We call 
 this assumption impartial scores culture (ISC) assumption. 
 \item We provide bounds for the fraction of manipulable profiles for voting rules
when the number of candidates is large. These bounds immediately tell us
whether or not the given voting rule is asymptotically collusion-proof. 
 We prove asymptotic results for plurality, veto, 
 and $k$-approval voting rules under both IC and ISC assumptions. 
 We prove Borda rule is not asymptotic strategy-proof on the number of candidates under 
 the IC assumption but the problem is still unresolved under the ISC assumption. 
 Our results show that the size of the coalition and the tie-breaking rule being 
 used play a crucial role in determining whether or not a voting rule is 
 asymptotically strategy-proof.
\end{itemize}

Our results on asymptotic collusion-proofness of voting rules 
are summarized as in Table 1. 
{\small
\begin{center}
\begin{table}[here]
\begin{minipage}{\textwidth}
 \begin{tabular}{|>{\centering\arraybackslash} m{1.6cm}|>{\centering\arraybackslash}m{1.9cm}|>{\centering\arraybackslash}m{2.5cm}|
 >{\centering\arraybackslash}m{0.7cm}|>{\centering\arraybackslash}m{2.6cm}|>{\centering\arraybackslash}m{1cm}|}\hline
  Cases& Plurality for\footnote{Plurality Rule with lexicographic tie breaking rule for the manipulators}	
  &Plurality against\footnote{Plurality Rule with lexicographic tie breaking rule against the manipulators}	
  & Veto		& $k$-Approval, $k>1$ 		& Borda\\\hline
  $c=1$, IC	& $\surd$		& $\times$		& $\times$	& $k=o(m):\times$		& $\times$\\\hline
  $c=1$	, ISC	& $\surd$		& $\times$		& $\times$	& $k=o(m):\times$		& $\times$\\\hline
  $c>1$	, IC	& $\times$		& $\times$		& $\times$	& $k=o(m):\times$		& $\times$\\\hline
  $c>1$	, ISC	& $\times$		& $\times$		& $\times$	& $k=o(m):\times$		& ?\\\hline
 \end{tabular}
\caption{Asymptotic strategy-proofness results on candidates}
\end{minipage}
\end{table}
\end{center}}

This work has been published before as an extended abstract in the international conference on Autonomous agents and multi-agent systems~\cite{dey2014asymptotic}.

\section{Preliminaries and New Definitions}\label{Preliminaries}
Let $\mathcal{V}=\{v_1, v_2, \dots, v_n\}$ be the set of \emph{voters} 
and $\mathcal{C}=\{c_1,c_2,\dots,c_m\}$ the set of \emph{candidates}.
Each voter $v_i$ has a \emph{preference} $\prec_i$ over the candidates which is a linear order over $\mathcal{C}$.
The set $\mathcal{L(C)}$ denotes the set of all linear orders over $\mathcal{C}$. 
A map $r_c:\uplus_{n\in\mathbb{N}^+}\mathcal{L(C)}^n\longrightarrow 2^\mathcal{C}\setminus\emptyset$
is called a \emph{voting correspondence}\footnote{By $\uplus$, we denote disjoint union. 
$\mathbb{N}^+ :=\{1, 2, 3, \dots\}$. $2^\mathcal{C}$ denotes the power set of $\mathcal{C}$. $\emptyset$ 
denotes the empty set.}. A map $t:2^\mathcal{C}\setminus\emptyset\longrightarrow \mathcal{C}$ is called a \emph{tie breaking rule}.
Commonly used tie breaking rules are \emph{lexicographic} tie breaking rule where ties are broken 
in accordance with a $\succ_t \in \mathcal{L(C)}$. We, in this 
paper, study two types of tie breaking rules - lexicographic tie breaking rule where ties 
are broken in favor of the manipulators and against the manipulators respectively. 
A \emph{voting rule} is $r=t\circ r_c$\footnote{$\circ$ denotes composition of mappings.}. 
For example, $\alpha=\left(\alpha_1,\alpha_2,\dots,\alpha_m\right)\in\mathbb{R}^m$ with $\alpha_1\ge\alpha_2\ge\dots\ge\alpha_m$ naturally defines a
voting rule - a candidate gets score $\alpha_i$ from a vote if it is placed at the $i^{th}$ position
and winner is the candidate with maximum score. This voting rule is called a \emph{positional scoring rule} based voting rule with score vector $\alpha$.
For $\alpha=\left(m-1,m-2,\dots,1,0\right)$, we get the \emph{Borda} voting rule. With $\alpha_i=1$ $\forall i\le k$ and $0$ else, 
the voting rule we get is known as $k$-\emph{Approval}. \emph{Plurality} is 1-\emph{Approval} and \emph{veto} is $(m-1)$-\emph{Approval}.

A voting rule is called \emph{anonymous} if the names of the voters are irrelevant. Formally, a 
voting rule $r$ is anonymous if
\[\forall n\in \mathbb{N}^+, \forall (\succ_1,\dots ,\succ_n)\in\mathcal{L(C)}^n, 
\forall \sigma \in \mathfrak{S}_n\footnote{$\mathfrak{S}_n$ \text{denotes the set of all permutations of the set} \{1, 2, \dots, n\}.}\]
\[r(\succ_1,\dots ,\succ_n) = r(\succ_{\sigma(1)},\dots ,\succ_{\sigma(n)})\]
A voting rule is called \textit{neutral} if the names of the candidates are immaterial. That is, 
\[\forall n\in \mathbb{N}^+, \forall (\succ_1,\dots ,\succ_n)\in\mathcal{L(C)}^n, \forall \sigma \in \mathfrak{S}_m\]
\[r(\sigma\circ\succ_1,\dots ,\sigma\circ\succ_n) = \sigma\circ r(\succ_1,\dots ,\succ_n)\]
A voting rule is called \emph{strategy-proof} if voters are not worse off by reporting 
true preferences than reporting any other preference. That is,
\[\forall n\in \mathbb{N}^+, \forall \succ^n\in\mathcal{L(C)}^n, \forall \succ,\succ^\prime\in\mathcal{L(C)},\]
\[r(\succ^n,\succ)\succ r(\succ^n,\succ^\prime)\]
A \emph{non-strategy-proof} voting rule is called \emph{manipulable}. 
The Gibbard and Satterthwaite Theorem 
says that, any unanimous, and non-dictatorial voting rule with at least three candidates is necessarily manipulable.

\subsection{Asymptotic Collusion-proofness and Strategy-proofness}\label{AsymptoticStrategyproofness}
Given a voting rule $r$, a set of candidates $\mathcal{C}$, $n$ 
number of truthful voters, and a coalition size $c$, we define 
$c$-collusion-proof voting profiles as follows.
\begin{definition}{(c-Collusion-proof Voting Profile)}\label{newdefCollusionProofness}
 A voting profile $\succ^n \in \mathcal{L(C)}^n$ is called \textit{c-collusion-proof} if $\forall \succ_1^c,\succ_2^c \in \mathcal{L(C)}^c$, 
 \[ r(\succ^n,\succ_1^c) \succ_1^c r(\succ^n,\succ_2^c) \footnote{\text{We say } $a (\succ_1,\dots ,\succ_n) b$ \text{ if } $a \succ_i b$ $\forall i=1, 2, \dots, n.$} \]
\end{definition}
For $c=1$, these profiles are called strategy-proof voting profiles. 
Given a voting rule $r$, we denote the set of all $c$-collusion-proof voting profiles by $T_r^c(\mathcal{C})$. The set of all 
strategy-proof voting profiles is denoted by $T_r(\mathcal{C})$. We will drop the subscript $r$ whenever the context is clear. 
Notice that manipulation by $c$ number of new voters is \textit{game theoretically} not possible in the above defined 
collusion-proof profiles. We define 
the above notions with respect to \textit{new} voters which directly implies that the hardness of the CM problem at collusion-proof instances 
is of no value. Previously, Slinko \cite{slinko2002asymptotic} defined strategy-proof 
voting profiles with respect to the existing voters. Obviously, our definition is more aligned with the formulation of the CM problem 
and hence the results give direct implications of the usefulness the CM problem being hard at a particular instance.

If $r$ is a score based voting rule\footnote{A voting rule is called a score based voting rule if winner is determined solely based on 
scores. Positional scoring rules are examples of such rules.}, then we define strategy-proof and collusion-proof scoring profiles as 
the scores that each candidate receive in some strategy-proof and collusion-proof voting profiles respectively. For example, $(8,2,2)$ 
is a scoring profile in a plurality election with $12$ voters and $3$ candidates. This is also a $5$-collusion-proof scoring profile 
since $5$ new voters can not manipulate this election. 
The set of all possible scoring profiles in an election with $n$ voters is denoted by $S_r^n(\mathcal{C})$. 
When the voting rule $r$ is clear from the context, we will drop $r$ from the notation.  

Clearly, since in score based voting rules, the winner is decided solely using the scores of the candidates, the above definitions are independent 
of the choice of strategy-proof or collusion-proof voting profiles if there are two or more voting profiles giving rise to the same 
scoring profile. 
The notion of collusion-proof scoring profiles has its motivation in computational social choice. For scoring rules, the CM 
problem can also be thought of as if, we are given 
a scoring profile generated out of the votes of the truthful voters and the question remains the same. Here 
collusion-proofness of a scoring profile determines whether the hardness of the CM problem is at all required or not at that profile. 
The above notions naturally lead us to study a society where all the scoring profiles are equally likely. We name this assumption the 
Impartial Scores Culture (ISC) assumption.

With the above definitions, the notions of \textit{asymptotic strategy-proofness} and \textit{asymptotic collusion-proofness} on voters and 
candidates are defined as follows. Let $\mathcal{D}$ be a probability distribution over the voting profiles.
\begin{definition}{(Asymptotic Strategy-proofness on Voters)}
 A voting rule r is called \textit{asymptotically strategy-proof on voters} if for all fixed and finite $\mathcal{C}$,
 \[ \lim_{n \to \infty} Prob_{\succ^n \sim \mathcal{D}} \{\succ^n \in T_r(\mathcal{C})\} = 1\]
\end{definition}
In words, a voting rule is called asymptotically strategy-proof on voters if \textit{almost all} the voting profiles are strategy-proof 
as we increase the number of voters. On similar lines, we define asymptotic strategy-proofness on candidates as follows. 
\begin{definition}{(Asymptotic Strategy-proofness on Candidates)}
 A voting rule r is called \textit{asymptotically strategy-proof on candidates} if $\exists N_0\in \mathbb{N}$ 
 such that $\forall n\ge N_0$,
 \[ \lim_{|\mathcal{C}| \to \infty} Prob_{\succ^n \sim \mathcal{D}} \{\succ^n \in T_r(\mathcal{C})\} = 1\]
\end{definition}
The above concepts are generalized to asymptotic collusion-proofness as follows. 
\begin{definition}{(Asymptotic c-Collusion-proofness on Voters)}
 A voting rule r is called \textit{asymptotically c-collusion-proof on voters} if for all fixed and finite $\mathcal{C}$,
 \[ \lim_{n \to \infty} Prob_{\succ^n \sim \mathcal{D}} \{\succ^n \in T_r^c(\mathcal{C})\} = 1\]
\end{definition}
\begin{definition}{(Asymptotic c-Collusion-proofness on Candidates)}
 A voting rule r is called \textit{asymptotically c-collusion-proof on candidates} if $\exists N_0\in \mathbb{N}$ 
 such that $\forall n\ge N_0$,
 \[ \lim_{|\mathcal{C}| \to \infty} Prob_{\succ^n \sim \mathcal{D}} \{\succ^n \in T_r^c(\mathcal{C})\} = 1\]
\end{definition}

\section{Results}

We show that with this \textit{modified} definition of 
strategy-proofness, many previously known results still hold. For almost 
all the voting rules studied in this paper, we show that they are not \textit{asymptotically strategy-proof on candidates}. 
That is, the problem of manipulation is really severe where computationally hard instances are. This shows that 
although there are evidences shown in the literature for computational hardness not being a very strong barrier, 
there is a good chance that the profiles which it is able to protect are really manipulable. We show the above results under both 
IC and ISC societal assumptions. IC assumption states that the voting profiles follow $\mathcal{U}(\mathcal{L(C)}^n)$ 
\footnote{Given a finite set $A$, $\mathcal{U}(A)$ denotes the uniform probability distribution over $A$.} 
probability distribution and the ISC assumes $\mathcal{U}(S_r^n(\mathcal{C}))$ probability distribution over 
scoring profiles. The following sections contain our results. In the interest of space, proof of some theorems 
have been skipped. They appear in the appendix.

\subsection{Plurality Voting Rule}

We first provide a complete characterization of the scoring profiles which are strategy-proof
under \emph{Plurality} voting rule.
\begin{observation}\label{manplural}
 For Plurality rule, following and only following  
 scoring profiles \\$(x_1,x_2,\dots,x_m)\in S^n([m])$ are strategy-proof:\footnote{$[m]:=\{1, 2, \dots, m\}$.}
 \begin{enumerate}
  \item
  \begin{enumerate}
   \item $|x_i-x_j|=1$ $or$ $0,\forall 1\le i,j\le m$ with lexicographic tie breaking for the manipulator.\label{almostequal}
   \item $x_1=x_2=\dots=x_m$ with lexicographic tie breaking against the manipulator.
  \end{enumerate}
  \item $\exists w\in [m]\text{ such that } x_w-x_i\ge2,\forall 1\le i\le m$, $i\neq w$, that is $c_w$ is the winner and 
  its score is greater than that of every other candidate by at least 2.
  \item $\exists w\in [m]\text{ such that } x_w> x_i,\forall 1\le i\le m$, $i\neq w$, with lexicographic tie breaking rule against the manipulator.
 \end{enumerate}
\end{observation}
In the profiles in 1(a), the scores of each pair of candidates differ by at most 1. 
We denote the set of all these scoring profiles by $E^n([m])$ for any scoring rule. The set of all corresponding 
voting profiles are denoted by $F^n([m])$. 
For $c$-collusion-proofness, we have the following characterization. 
\begin{observation}\label{manpluralcoalitional}
 For plurality rule, following and only following scoring profiles $(x_1,x_2,\dots,x_m)\in S^n([m])$ are c-collusion (c$>$1) proof:\\
  $\exists w\in\{1,2,\dots,m\}\text{ such that } x_w-x_i\ge c+1,\forall 1\le i\le m$ $i\neq w$ , that is $c_w$ is the winner and 
  its score is at least c+1 greater than that of every other candidate.
\end{observation}
The following lemma provides the cardinality of $|S^n([m])|$ for plurality election. 
\begin{lemma}\label{pluralityscorecard}
 For Plurality rule, $|S^n([m])| = {n+m-1 \choose m-1} $.
\end{lemma}
\begin{theorem}
 Let A be the set of all $c$-collusion-proof Plurality scoring profiles. Then $\frac{|A|}{|S^n([m])|} \ge (\frac{n-c}{n+1})^{m-1}$.
\end{theorem}
\begin{proof}
 Let $c$ be the coalition size; $c$ = $o(n)$. Consider the following map,
 \[f:S^{n-c-1}([m])\longrightarrow A, x:=(x_1, \dots, x_m)\mapsto (x_1, \dots,x_w+c+1,\dots,x_m)\]
 $c_w$ is the winner of scoring profile $x$. $f$ is injective. Hence from Lemma\nobreakspace \ref {pluralityscorecard}, \\
 \[|A|\ge|S^{n-c-1}([m])|={m+n-c-2\choose m-1}\]
 Now we have, 
 \begin{eqnarray*}
 \frac{|A|}{|S^n([m])|} &\ge& \frac{{m+n-c-2\choose m-1}}{{m+n-1\choose m-1}}\\
			&=& \frac{(m+n-c-2)\dots(n-c)}{(m+n-1)\dots(n+1)}\\
			&\ge& \left(\frac{n-c}{n+1}\right)^{m-1}
 \end{eqnarray*}
The last inequality comes from the fact that $\frac{a}{b} \ge \frac{a-1}{b-1}, \forall a,b \in \mathbb{N}, 0<a<b$. \qed
\end{proof}
The above bound immediately proves that plurality rule under ISC assumption is asymptotically strategy-proof and 
 asymptotically $o(n)$-collusion-proof on the number of voters.
The above results imply that in elections where the number of voters, is large the problem of manipulation is 
not severe. 
\begin{theorem}\label{pluralityEBound}
 For Plurality rule, for $m>n$, $\frac{|E^n([m])|}{|S^n([m])|} \ge \left(\frac{m-n+1}{m}\right)^n$.
\end{theorem}
 Hence plurality rule under ISC assumption is asymptotically strategy-proof on candidates for 
 lexicographic tie breaking rule for the manipulator, but not asymptotically strategy-proof on
 candidates with lexicographic tie breaking rule against manipulator. Plurality is not 
 $c$-collusion-proof for $c>1$ on candidates.

\subsection{$k$-Approval Voting Rule}

$k$-Approval voting rule is the scoring rule with the score vector $(1,\dots,1,0,\dots,0)$, the first $k$ entries are 1 and rest are 0. 
The scoring profiles are given by,

$S^n([m]):=\{(x_1,\dots,x_m)\in\mathbb{N}^m:\sum_{i=1}^m x_i=nk, 0\le x_i\le n, \forall i \in [m]\}$\\
$|S^n([m])|= \text{Coefficient of }x^{nk} \text{ in }(1+x+x^2+\dots+x^n)^m = 
\sum_{i=0}^k (-1)^i{n(k-i)+m-1\choose m-1}$
\begin{observation}\label{mankapproval}
 For k-Approval rule with $1<k<m$, following and only following scoring profiles $(x_1,x_2,\dots,x_m)\in S^n([m])$ are strategy-proof:
 \begin{enumerate}
  \item $\exists w\in\{1,2,\dots,m\}\text{ such that } x_w-x_i\ge2,\forall 1\le i\le m,i\neq w,$ that is $c_w$ 
  is the winner and its score is at least two more than every other candidates.
  \item $x_1=x_2=\dots=x_m$ with lexicographic tie breaking rule for the manipulator.
 \end{enumerate}
\end{observation}
For $k$-Approval voting rule with $1<k<m$, we have the following results.
\begin{theorem}\label{kApprovalFBound}
 For $k$-Approval voting rule, $\frac{|F^n([m])|}{|T^n([m])|} \ge \left(\frac{m-nk+1}{m}\right)^{nk}$.
\end{theorem}
\begin{proof}
We have,
\begin{eqnarray*}
 |F^n([m])| &=& |E^n([m])|(nk)!\left(\left(m-k\right)!\right)^n\\ 
	    &=& {m \choose nk}(nk)!\left(\left(m-k\right)!\right)^n\\ 
	    &=& \frac{m!\left(\left(m-k\right)!\right)^n}{\left(m-nk\right)!}
\end{eqnarray*}
The first equality comes from the fact that, for each scoring profile in $E^n([m])$, there are 
$(nk)!\left(\left(m-k\right)!\right)^n$ many voting profiles which gives that scoring profile.
\begin{eqnarray*}
 \frac{|F^n([m])|}{|T^n([m])|}  &=&   \frac{m!}{\left(m-nk\right)!\left(m(m-1)\dots(m-k+1)\right)^n} \\
				&=&   \frac{m(m-1)\dots(m-nk+1)}{\left(m-nk\right)!\left(m(m-1)\dots(m-k+1)\right)^n} \\
				&\ge& \left(\frac{m-nk+1}{m}\right)^{nk}
\end{eqnarray*} \qed
\end{proof}
 Hence plurality rule under IC assumption is asymptotically strategy-proof on candidates for 
 lexicographic tie breaking rule for the manipulator, but asymptotically manipulable on
 candidates with lexicographic tie breaking rule against manipulator. Plurality is not 
 $c$-collusion-proof for $c>1$ on candidates.
 For $k>1, k=o(m)$, k-Approval is not asymptotically strategy-proof on number of candidates.
\begin{observation}\label{mankapprovalcoalition}
 For k-Approval rule with $1<k<m$, following scoring profiles $(x_1,x_2,\dots,x_m)\in S^n([m])$ are c-collusion-proof:
 \begin{enumerate}
  \item $\exists w\in\{1,2,\dots,m\}\text{ such that } x_w-x_i\ge2c,\forall 1\le i\le m,i\neq w,$ that is $c_w$ 
  is the winner and its score is at least two more than every other candidates.
 \end{enumerate}
\end{observation}
\begin{theorem}\label{kapprovalAsymptoticSCn}
 For $1<k<m$, $k$-approval is asymptotically $o(n)$-collusion-proof on the number of voters under ISC assumption.
\end{theorem}

\subsubsection*{Veto Rule : }

Veto voting rule is $(m-1)$-Approval voting rule. Now since scoring rules are invariant under affine transformations,
we can consider veto as a scoring rule with score vector being $(0,0,\dots,0,-1)$. Clearly for veto,
\[S^n([m])=\{(x_1,x_2,\dots,x_m)\in\mathbb{Z}_{\le0}^m:\sum_{i=1}^m x_i=-n\}\]
Hence $|S^n_{veto}([m])|=|S^n_{plurality}([m])|={n+m-1\choose n}$.
Theorem\nobreakspace \ref {kapprovalAsymptoticSCn} implies that 
veto is asymptotically $o(n)$-collusion-proof on the number of voters under ISC assumption.
\begin{observation}\label{manveto}
 For veto rule following scoring profiles $(x_1,x_2,\dots,x_m)\in S^n([m])$ are $c$-collusion-proof:
 \begin{enumerate}
  \item $\exists w\in\{1,2,\dots,m\}\text{ such that } x_w-x_i\ge c+1,\forall 1\le i\le m$ $i\neq w$ , i.e. $c_w$ is 
  the winner and its score is at least two more than every other candidates.
  \item For $c<m-1,$ $x_1=x_2=\dots=x_m$.
 \end{enumerate}
\end{observation}
\begin{lemma}
 For Veto rule, for $m>n$, $\frac{|E^n([m])|}{|S^n([m])|} \ge \left(\frac{m-n+1}{m}\right)^n$.
\end{lemma}
 Hence veto rule under ISC assumption is not asymptotically strategy-proof on candidates.
\begin{lemma}
 For Veto rule, for $m>n$, $\frac{|F^n([m])|}{|S^n([m])|} \ge \left(\frac{m-n+1}{m}\right)^n$.
\end{lemma}
 The above lemma shows that veto rule under IC assumption is not asymptotically strategy-proof on candidates.
The following theorem bounds the number of $c$-collusion-proof strategy profiles for veto voting rule which 
will subsequently help us to prove $o(n)$-collusion-proofness of veto voting rule. 
\begin{theorem}
 Let A be the set of all $c$-collusion-proof Veto scoring profiles. 
 Then $\frac{|A|}{|S^n([m])|} \ge \left(\frac{n-(c+1)(m-1)+1}{mn-(c+1)(m-1)}\right)^{(c+1)(m-1)}$.
\end{theorem}
 So under ISC , Veto voting rule is asymptotically $o(n)$-collusion-proof on the number of voters.

\subsection{Borda Voting Rule}
For Borda voting rule, a score vector is $(m-1,m-2,\dots,0)$.
\begin{observation}\label{manborda}
 For Borda rule, following and only following scoring profiles $(x_1,x_2,\dots,x_m)\in S^n([m])$ are strategy-proof:
 \begin{enumerate}
  \item $\exists w\in\{1,2,\dots,m\}\text{ such that } x_w-x_i\ge m,\forall 1\le i\le m,i\neq w,$ that is $c_w$ is the winner and its score is at 
  least m greater than that of every other candidates.
  \item $x_1=x_2=\dots=x_m$ \label{bordaequal}
  \item Every candidate except one is tied with highest score and the loser's score is one less. \label{bordatied}
 \end{enumerate}
\end{observation}
Before going to the asymptotic result for the Borda voting rule, let us state some lemmas which we will be using later.
\begin{lemma}\label{ex}
 $\sum_{i=2}^{\infty} (-1)^i (i-1) \frac{x^i}{i!} = 1 - (1+x)e^{-x}, \forall x\in \mathbb{R}$.
\end{lemma}
\begin{lemma}\label{sum}
 $\forall l \in \mathbb{N}$
 $1 - (-1)^l (l-1) = \sum_{i=2}^{l-1} (-1)^i (i-1) {l \choose i}$.
\end{lemma}
The next lemma gives a useful cardinality formula. Suppose we have $m$ finite sets namely $A_1, \dots, A_m$. Let $A$ 
be the set of all elements which are present in at least two sets. Then the following lemma computes the cardinality 
of the set $A$.
\begin{lemma}\label{twoelementset}
 $|A| = \sum_{r=2}^{m} (-1)^r (r-1) \sum_{I \in {[m] \choose r}} |A_I|
 \footnote{Let A be a set. Then ${A \choose r}:=\{B \subseteq A : |B| = r\}$.
 I be any subset of an index set. Then $A_I := \cap_{i \in I} A_i$.}$.
\end{lemma}
The following lemma gives a sufficient condition for a voting profile to be manipulable for large number of candidates.
\begin{lemma}\label{Bordasufficient}
 For fixed $n$ and \textit{large}\footnote{large enough such that Item\nobreakspace \ref {bordaequal} and Item\nobreakspace \ref {bordatied} of the \MakeUppercase proposition\nobreakspace \ref {manborda} cannot occur.} $m$, 
 if in a voting profile $(\succ_1, \dots, \succ_n) \in \mathcal{L(C)}^n$, 
 $\exists c_i,c_j \in \mathcal{C}, c_i \ne c_j$ such that both $c_i$ and $c_j$ are ranked in the top 
 $l$ positions in each preference $\succ_i, 1\le i \le n$ where $(l-1)n<m$, then the profile is manipulable.
\end{lemma}
\begin{theorem}\label{BordaAsymptoticSTmm}
 Borda rule is not asymptotically strategy-proof on the number of candidates under IC assumption.
\end{theorem}
\begin{proof}
 Consider an election with $n$ voters, $m$ candidates. Let us define 
 $l := \lambda m^{\frac{n-1}{n}}, \lambda \in \mathbb{N}$ where $\lambda$ is a constant.  Let us 
 assume $m$ to be large enough to satisfy the setting of Lemma\nobreakspace \ref {Bordasufficient}. The existence of such an 
 $m$ is guaranteed from the fact that $l = \Theta(m^{\frac{n-1}{n}})$. 
 Let us define $A_i$ for $i\in[m]$ as follows, \\
 $A_i := \{(\succ_j)_{j=1}^n \in \mathcal{L(C)}^n : 
 i \text{ is ranked in the top } l \text{ positions in } \succ_j, 1\le j \le n\}$\\
 Clearly,
 $ I \subseteq [m], |I| = k $ then, 
 $$ |A_I| = 
  \left\{
  \begin{array}{lr}
   \left( \left( m-l \right)!\Pi_{i=0}^{k-1} \left(l-i\right) \right)^n &, \text{ if } k \le l\\
   0									&, \text{ if } k >l
  \end{array}
  \right.
 $$
 Let us define the set $A$ as follows,
 $$ A := \{ \succ \in \mathcal{L(C)}^n : \exists i, j \in [m], i \ne j, \succ \in A_i, \succ \in A_j\} $$
 Let us define the set of all manipulable voting profiles $M^n([m])$ as 
 $$M^n([m]) := \mathcal{L}([m])^n \setminus T^n([m])$$
 From Lemma\nobreakspace \ref {Bordasufficient}, it follows that $ A \subseteq M^n([m]) $. Now,
 
\begin{eqnarray*}
 \frac{|M^n([m])|}{|\mathcal{L(C)}^n|}	&\ge&	\frac{|A|}{(m!)^n}\\
						&=&	\sum_{r=2}^{l} (-1)^r (r-1) \sum_{I \in {[m] \choose r}} \frac{|A_I|}{(m!)^n}\\
						&=&	\sum_{r=2}^{l} (-1)^r (r-1) {m \choose r} \left(\Pi_{i=0}^{r-1} \frac{l-i}{m-i}\right)^n\\
   \lim_{m \to \infty}	\frac{|M^n([m])|}{|\mathcal{L(C)}^n|}
						&\ge&	\sum_{r=2}^{\infty} (-1)^r (r-1) \frac{\lambda^{rn}}{r!}\\
						&=&	1 - (1+\lambda^n)e^{-\lambda^n}
\end{eqnarray*}
 The second step follows from Lemma\nobreakspace \ref {twoelementset}, the fourth step is a result of using the fact that $l = \Theta(m^{\frac{n-1}{n}})$, 
 and the last step follows from Lemma\nobreakspace \ref {ex}. Now since the above 
 inequality is true for all $\lambda \in \mathbb{R}$, we get,
 $$ \lim_{m \to \infty}	\frac{|M^n([m])|}{|\mathcal{L(C)}^n|} \ge \sup_{\lambda \in \mathbb{N}} \{1 - (1+\lambda^n)e^{-\lambda^n}\} = 1 $$
 Therefore the Borda rule is not asymptotically strategy-proof on the number of candidates under IC assumption. \qed
\end{proof}
The above theorem shows that it is highly likely that the voting profiles for which manipulating the Borda 
voting rule is intractable are actually manipulable. Hence computational hardness provides some resistance against 
manipulation of the Borda rule. The behavior of the Borda rule under ISC assumption is still open. 

\section{Discussion and Future Work}\label{ConclusionsFutureWork}
For all the voting rules studied here, the severity of manipulation has been shown to be 
present mostly in the elections with large number of candidates and not in the cases with many voters. 
Also this is the first paper, to the best of our knowledge, which tries to address the reverse 
question posed in the beginning - the profiles which are being protected by computational 
hardness are at all manipulable or not. We have studied some popular voting rules in this paper. 
Certainly a general result would be much more interesting which is a prospective 
future work.

We have explored a new societal assumption called ISC and argued for 
its need in the context of computational social choice. We notice a positive correlation among 
results under ISC assumption and under IC assumption. Again a general result 
connecting these two concepts is another important future direction of research. Here we 
conjecture that a voting procedure is asymptotically strategy-proof or $c$-collusion 
proof either on the number of voters or on the number of candidates under ISC 
assumption if and only if it is so under IC assumption.
An interesting future direction is to find axiomatic characterization of 
asymptotically strategy-proof voting rules. 

\bibliographystyle{apalike}
\bibliography{asymptotic}

\begin{thebibliography}{}

\bibitem[Bartholdi and Orlin, 1991]{bartholdi1991single}
Bartholdi, J. and Orlin, J. (1991).
\newblock Single transferable vote resists strategic voting.
\newblock {\em Social Choice and Welfare}, 8(4):341--354.

\bibitem[Bartholdi et~al., 1989]{bartholdi1989computational}
Bartholdi, J., Tovey, C., and Trick, M. (1989).
\newblock The computational difficulty of manipulating an election.
\newblock {\em Social Choice and Welfare}, 6(3):227--241.

\bibitem[Dey and Narahari, 2014]{dey2014asymptotic}
Dey, P. and Narahari, Y. (2014).
\newblock Asymptotic collusion-proofness of voting rules: the case of large
  number of candidates.
\newblock In {\em Proceedings of the 2014 international conference on
  Autonomous agents and multi-agent systems}, pages 1419--1420. International
  Foundation for Autonomous Agents and Multiagent Systems.

\bibitem[Ephrati and Rosenschein, 1991]{ephrati1991clarke}
Ephrati, E. and Rosenschein, J. (1991).
\newblock The {C}larke tax as a consensus mechanism among automated agents.
\newblock In {\em Proceedings of the Ninth National Conference on Artificial
  Intelligence ({AAAI})}, pages 173--178.

\bibitem[Faliszewski et~al., 2008]{faliszewski2008copeland}
Faliszewski, P., Hemaspaandra, E., and Schnoor, H. (2008).
\newblock Copeland voting: Ties matter.
\newblock In {\em Proceedings of the 7th International Conference on Autonomous
  Agents and Multiagent Systems- AAMAS-2008,Volume 2}, pages 983--990.

\bibitem[Gibbard, 1973]{gibbard1973manipulation}
Gibbard, A. (1973).
\newblock Manipulation of voting schemes: a general result.
\newblock {\em Econometrica: Journal of the Econometric Society}, pages
  587--601.

\bibitem[Nitzan, 1985]{nitzan1985vulnerability}
Nitzan, S. (1985).
\newblock The vulnerability of point-voting schemes to preference variation and
  strategic manipulation.
\newblock {\em Public Choice}, 47(2):349--370.

\bibitem[Pattanaik, 1975]{pattanaik1975strategic}
Pattanaik, P.~K. (1975).
\newblock Strategic voting without collusion under binary and democratic group
  decision rules.
\newblock {\em The Review of Economic Studies}, pages 93--103.

\bibitem[Pennock et~al., 2000]{pennock2000social}
Pennock, D., Horvitz, E., Giles, C., et~al. (2000).
\newblock Social choice theory and recommender systems: Analysis of the
  axiomatic foundations of collaborative filtering.
\newblock In {\em Proceedings of the National Conference on Artificial
  Intelligence}, pages 729--734. Menlo Park, CA; Cambridge, MA; London; AAAI
  Press; MIT Press; 1999.

\bibitem[Procaccia and Rosenschein, 2007]{procaccia2007average}
Procaccia, A. and Rosenschein, J. (2007).
\newblock Average-case tractability of manipulation in voting via the fraction
  of manipulators.
\newblock In {\em Proceedings of the 7th International Conference on Autonomous
  Agents and Multiagent Systems- AAMAS-2007}, volume~7.

\bibitem[Satterthwaite, 1975]{satterthwaite1975strategy}
Satterthwaite, M. (1975).
\newblock Strategy-proofness and arrow's conditions: Existence and
  correspondence theorems for voting procedures and social welfare functions.
\newblock {\em Journal of Economic Theory}, 10(2):187--217.

\bibitem[Slinko, 2002]{slinko2002asymptotic}
Slinko, A. (2002).
\newblock On asymptotic strategy-proofness of the plurality and the run-off
  rules.
\newblock {\em Social Choice and Welfare}, 19(2):313--324.

\bibitem[Xia et~al., 2010]{xia2010scheduling}
Xia, L., Conitzer, V., and Procaccia, A. (2010).
\newblock A scheduling approach to coalitional manipulation.
\newblock In {\em Proceedings of the 11th ACM Conference on Electronic
  Commerce}, pages 275--284. ACM.

\end{thebibliography}

\newpage
\appendix
\section{Appendix}

\setcounter{observation}{0}

\begin{observation}
 For Plurality rule, following and following 
 scoring profiles \\$(x_1,x_2,\dots,x_m)\in S^n([m])$ are strategy-proof: 
 \begin{enumerate}
  \item
  \begin{enumerate}
   \item $|x_i-x_j|=1$ $or$ $0,\forall 1\le i,j\le m$ for lexicographic tie breaking for the manipulator.
   \item $x_1=x_2=\dots=x_m$ for lexicographic tie breaking against the manipulator.
  \end{enumerate}
  \item $\exists w\in [m]\text{ such that } x_w-x_i\ge2,\forall 1\le i\le m$, $i\neq w$, that is $c_w$ is the winner and 
  its score is at least two greater than that of every other candidate.
  \item $\exists w\in [m]\text{ such that } x_w> x_i,\forall 1\le i\le m$, $i\neq w$, with lexicographic tie breaking rule against the manipulator.
 \end{enumerate}
\end{observation}
\begin{proof}
The first case talks about the situations where the score of all the candidates are as much same as possible. That is some candidates have 
highest score each and others have score just one less than highest. Now if the new voter's most preferred candidate is already a candidate with highest 
score then her vote will make the candidate the winner. If her most preferred candidate has not scored highest then her vote will 
results in a tie and hence will make its candidate win only with lexicographic tie breaking rule for the manipulators. Hence the first case is strategy-proof.

The second and third cases deal with non-pivotal profiles, that is the new candidate cannot change current winner and hence she has no incentive to lie. 
Thus the scoring profile is strategy-proof.

In all other cases, the number of candidates tied for win or has score one less than the winner's, is more than one and there is a candidate whose score 
is at least two less than the winner's score. If the voter's most preferred candidate is the candidate with least score then she cannot make 
her most preferred candidate win. But she can make one of the tied candidates and runner ups a winner depending upon the tie breaking rule being used. 
Hence these scoring profiles are manipulable. \qed
\end{proof}
\begin{observation}
 For plurality rule, following and only following scoring profiles $(x_1,x_2,\dots,x_m)\in S^n([m])$ are c-collusion (c$>$1) proof:
 \begin{itemize}
  \item $\exists w\in\{1,2,\dots,m\}\text{ such that } x_w-x_i\ge c+1,\forall 1\le i\le m$ $i\neq w$ , that is $c_w$ is the winner and 
  its score is at least c+1 more than every other candidate.
 \end{itemize}
\end{observation}
\begin{proof}
 On similar line of the proof of the \MakeUppercase proposition\nobreakspace \ref {manplural}.
\end{proof}
\setcounter{lemma}{0}
\setcounter{theorem}{1}
\begin{lemma}
 For Plurality rule, $|S^n([m])| = {n+m-1 \choose m-1} $.
\end{lemma}
\begin{proof}
 The set $S^n([m])$ is the solution set of the following equation.
 \[ \sum_{i=1}^m x_i = n, x_i \in \{ 0, 1, \dots, n\},\forall i \in [m] \]
 Hence $|S^n([m])| = {n+m-1 \choose m-1}$. \qed
\end{proof}
\begin{theorem}
 For Plurality rule, for $m>n$, $\frac{|E^n([m])|}{|S^n([m])|} \ge \left(\frac{m-n+1}{m}\right)^n$.
\end{theorem}
\begin{proof}
We have,
\begin{eqnarray*}
 \frac{|E^n([m])|}{|S^n([m])|}  &=&   \frac{{m \choose n}}{{m+n-1 \choose n}} \\
				&=&   \frac{m}{m+n-1}.\frac{m-1}{m+n-2}\dots \frac{m-n+1}{m} \\
				&\ge& \left(\frac{m-n+1}{m}\right)^n
\end{eqnarray*}
The last inequality comes from the fact that $\frac{a}{b} \ge \frac{a-1}{b-1}, \forall a,b \in \mathbb{N}, 0<a<b$. \qed
\end{proof}
\begin{observation}
 For k-Approval rule with $1<k<m$, following and only following scoring profiles $(x_1,x_2,\dots,x_m)\in S^n([m])$ are strategy-proof:
 \begin{enumerate}
  \item $\exists w\in\{1,2,\dots,m\}\text{ such that } x_w-x_i\ge2,\forall 1\le i\le m,i\neq w,$ that is $c_w$ 
  is the winner and its score is at least two more than every other candidates.
  \item $x_1=x_2=\dots=x_m$ with lexicographic tie breaking rule for the manipulator.
 \end{enumerate}
\end{observation}
\begin{proof}
 The first case deals with non-pivotal cases, i.e. the new candidate can not change current winner and hence he/she has no incentive to lie. Thus the scoring profile in
non-manipulable. In the second case, the situation is same as with one voter and thus non-manipulable.\\
In all other cases, let $W=\{c_1,\dots,c_w\}$ be the set of all candidates with maximum score, $R=\{c_{w+1},\dots,c_m\}$ be rest of the candidates. 
If the new voter's true preference is $c_1\succ c_2\succ \dots c_w\succ c_{w+1}\succ\dots c_m$ then new voter can clearly manipulate. \qed
\end{proof}
\begin{observation}
 For k-Approval rule with $1<k<m$, following scoring profiles $(x_1,x_2,\dots,x_m)\in S^n([m])$ are c-collusion-proof:
 \begin{enumerate}
  \item $\exists w\in\{1,2,\dots,m\}\text{ such that } x_w-x_i\ge2c,\forall 1\le i\le m,i\neq w,$ that is $c_w$ 
  is the winner and its score is at least two more than every other candidates.
 \end{enumerate}
\end{observation}
\begin{proof}
 The new $c$ voters cannot change the outcome of the election and hence the profiles are $c$-collusion-proof. \qed
\end{proof}
\setcounter{theorem}{3}
\begin{theorem}
 For $1<k<m$, $k$-approval is asymptotically $o(n)$-collusion-proof on the number of voters under ISC assumption.
\end{theorem}
\begin{proof}
 Define $l:={m-1\choose k-1}, l^\prime:={m-2\choose k-2}$. Let $c$ be the coalition size. $c=o(n)$. 
 Let A be the set of all $c$-collusion-proof scoring profiles with $n$ voters and $m$ candidates.
 Consider the following map,\\ 
 \[f:S^{n-lc}([m])\longrightarrow A\]
 \[(x_1,x_2,\dots,x_m)\mapsto (x_1+l^\prime,\dots,x_w+l,\dots,x_m+l^\prime)\] 
 where $c_w$ is the winner. Hence,\\ 
 \begin{eqnarray*}
  |A| &\ge& |S^{n-lc}([m])|\\
      &=& \sum_{i=0}^k (-1)^i{(n-lc)(k-i)+m-1\choose m-1}
 \end{eqnarray*}
 Now,
 \begin{eqnarray*}
  1 &\ge& \lim_{n \to \infty} \frac{|A|}{|S^n([m])|}\\
    &\ge& \lim_{n \to \infty} \frac{\sum_{i=0}^k (-1)^i{(n-lc)(k-i)+m-1\choose m-1}}{\sum_{i=0}^k (-1)^i{n(k-i)+m-1\choose m-1}}\\
    &=& 1
 \end{eqnarray*}
 The last equality comes from the fact that $c=o(n)$, and $l,l^{\prime}$ are independent of $n$. Hence the probability 
 that an uniformly randomly chosen scoring profile is $o(n)$-collusion-proof goes to one as the number of voters increases 
 and the result follows. \qed
\end{proof}
\begin{observation}
 For veto rule following scoring profiles $(x_1,x_2,\dots,x_m)\in S^n([m])$ are $c$-collusion-proof:
 \begin{enumerate}
  \item $\exists w\in\{1,2,\dots,m\}\text{ such that } x_w-x_i\ge c+1,\forall 1\le i\le m$ $i\neq w$ , i.e. $c_w$ is 
  the winner and its score is at least two more than every other candidates.
  \item For $c<m-1,$ $x_1=x_2=\dots=x_m$.
 \end{enumerate}
\end{observation}
\begin{proof}
 The first case deals with non-pivotal cases, i.e. the new voter can not change current winner and hence she has no incentive to lie. Thus the scoring profile in
non-manipulable.\\
In the second case, the situation is same as with one voter and thus non-manipulable.\\
In all other cases, let $W$ be the set of all candidates with maximum score, $R$ be the set of all candidates with score one less than winners' and $L$ be the set of rest
of the candidates. Not $W\ne\mathcal{C}$ since otherwise this will be same as case two. Hence at least one of $R$ or $L$ must be non-empty. Now if the 
least preferred candidate of new voters is not in $W$ then new voters can certainly manipulate. \qed
\end{proof}
\begin{lemma}
 For Veto rule, for $m>n$, $\frac{|E^n([m])|}{|S^n([m])|} \ge \left(\frac{m-n+1}{m}\right)^n$.
\end{lemma}
\begin{proof} 
 Consider the following map,\\ 
 \[f:E^n_{(veto)}([m])\longrightarrow E^n_{(plurality)}([m])\]
 \[(x_1,x_2,\dots,x_m)\mapsto(-x_1,-x_2,\dots,-x_m)\]
 The map is clearly bijective and hence the result follows from Theorem\nobreakspace \ref {pluralityEBound}. \qed
\end{proof}
\begin{lemma}
 For Veto rule, for $m>n$, $\frac{|F^n([m])|}{|S^n([m])|} \ge \left(\frac{m-n+1}{m}\right)^n$.
\end{lemma}
\begin{proof} 
 Follows from the fact that $|F^n_{(veto)}([m])| = |F^n_{(plurality)}([m])|$ and Theorem\nobreakspace \ref {kApprovalFBound}. \qed
\end{proof}
\setcounter{theorem}{4}
\begin{theorem}
 Let A be the set of all $c$-collusion-proof Veto scoring profiles. 
 Then $\frac{|A|}{|S^n([m])|} \ge \left(\frac{n-(c+1)(m-1)+1}{mn-(c+1)(m-1)}\right)^{(c+1)(m-1)}$.
\end{theorem}
\begin{proof}
 Let $c$ be the coalition size. Consider the following map,
 \[f:S^{n-(c+1)(m-1)}([m])\longrightarrow A\]
 \[x=(x_1, \dots, x_m)\mapsto (x_1-(c+1), \dots,x_w,\dots,x_m-(c+1))\] 
 $c_w$ is the winner of scoring profile $x$. This map is clearly injective. Hence \\
 $|A|\ge|S^{n-(c+1)(m-1)}([m])|={m+n-(c+1)(m-1)-1\choose m-1}$. We have, 
 \begin{eqnarray*}
 \frac{|A|}{|S^n([m])|} &\ge& \frac{{m+n-(c+1)(m-1)-1\choose m-1}}{{m+n-1\choose m-1}}\\
			&=& \frac{n\dots(n-(c+1)(m-1)-1)}{(m+n-1)\dots(m+n-(c+1)(m-1))}\\
			&\ge& \left(\frac{n-(c+1)(m-1)+1}{mn-(c+1)(m-1)}\right)^{(c+1)(m-1)}
 \end{eqnarray*}
The last inequality comes from the fact that $\frac{a}{b} \ge \frac{a-1}{b-1}, \forall a,b \in \mathbb{N}, 0<a<b$. \qed
\end{proof}
\begin{observation}
 For Borda rule, following and only following scoring profiles $(x_1,x_2,\dots,x_m)\in S^n([m])$ are strategy-proof:
 \begin{enumerate}
  \item $\exists w\in\{1,2,\dots,m\}\text{ such that } x_w-x_i\ge m,\forall 1\le i\le m,i\neq w,$ that is $c_w$ is the winner and its score is at 
  least m more than every other candidates.
  \item $x_1=x_2=\dots=x_m$ 
  \item Every candidate except one are tied with highest score and the loser's score is one less. 
 \end{enumerate}
\end{observation}
\begin{proof}
 First case covers non-pivotal scoring profiles, i.e. current winner does not change irrespective of new voter's vote. The second case is similar 
 to one voter case. In the third case, it is always best response for a new voter to cast his/her true preference. Hence these are 
 non-manipulable scoring profiles. Now we will show that every other scoring profiles are manipulable.\\ 
 \emph{Case I - unique candidate with highest score:} In this case, runner up candidate's score is within m of winner's score. Let $c_w$ is the winner and $c_r$ is a runner up.
 Consider a new voter with preference $c_r\succ c_w\succ\dots$. Now if the voter casts truthfully then $c_w$ will be a winner. But casting 
 $c_r\succ\dots\succ c_w$ makes $c_r$ unique winner. Hence this scoring profile is manipulable.\\ 
 \emph{Case II - more than one candidate with highest score:} In this case more than one candidate has got highest score. Without loss of generality we may assume $c_1,c_2,\dots,c_i$ with $1<i<m$. 
 If $i<m-1$ then a new voter  with preference $c_m\succ c_1\succ c_2\succ\dots\succ c_i\succ\dots$ is better off by casting 
 $c_m\succ\dots\succ c_1\succ c_2\succ\dots\succ c_i$. If $i=m-1$ then $x_m<x_1-1$ and hence a  voter with preference 
 $c_m\succ c_1\succ c_2\succ\dots\succ c_{m-1}$ is better off by reporting $c_1\succ c_2\succ\dots\succ c_m$. Hence these profiles are manipulable. \qed
\end{proof}
\setcounter{lemma}{3}
\begin{lemma}
 $\sum_{i=2}^{\infty} (-1)^i (i-1) \frac{x^i}{i!} = 1 - (1+x)e^{-x}$.
\end{lemma}
\begin{proof}
We know,
 $$
  \begin{array}{rclcl}
   && 			e^{-x}   &=& \sum_{i=0}^{\infty} (-1)^i \frac{x^i}{i!}\\
   &\Rightarrow&	-e^{-x}  &=& \sum_{i=0}^{\infty} (-1)^{i-1} \frac{(i+1)x^i}{(i+1)!}\\
   &\Rightarrow&	-xe^{-x} &=& \sum_{i=1}^{\infty} (-1)^{i} \frac{ix^{i}}{i!}\\
   &\Rightarrow&	-xe^{-x} &=& \sum_{i=2}^{\infty} (-1)^i (i-1) \frac{x^i}{i!}\\
   &&				 &&  + \sum_{i=1}^{\infty} (-1)^i \frac{x^i}{i!}\\
   &\Rightarrow&	-xe^{-x} &=& \sum_{i=2}^{\infty} (-1)^i (i-1) \frac{x^i}{i!} + e^{-x} -1
  \end{array}
 $$
 The second step follows by taking derivative with respect to $x$ on both side. The result follows from the last step. \qed
\end{proof}
\begin{lemma}
 $\forall l \in \mathbb{N}$
 $$1 - (-1)^l (l-1) = \sum_{i=2}^{l-1} (-1)^i (i-1) {l \choose i}.$$
\end{lemma}
\begin{proof}
 We know that $\forall l \in \mathbb{N}, x \in \mathbb{R}^+$,
 $$
 \begin{array}{rclcl}
  &&		(1-x)^l 						&=&	\sum_{i=0}^{l} (-1)^l {l \choose i} x^i\\
  &\Rightarrow&	\frac{(1-x)^l}{x} 					&=&	\sum_{i=0}^{l} (-1)^l {l \choose i} x^{i-1}\\
  &\Rightarrow&	-\frac{(1-x)^l}{x^2} - l\frac{(1-x)^{l-1}}{x} 		&=&	\sum_{i=0}^{l} (i-1) (-1)^l {l \choose i} x^{i-2}\\
  &\Rightarrow&	1 - (-1)^l (l-1)				 	&=&	 \sum_{i=2}^{l-1} (-1)^i (i-1) {l \choose i}
 \end{array}
 $$
 The second step is derived by dividing both side by $x$. At third step, we take derivative with respect to $x$. We have 
 instantiated $x$ to be 1 at fourth step. \qed
\end{proof}
\begin{lemma}
 $|A| = \sum_{r=2}^{m} (-1)^r (r-1) \sum_{I \in {[m] \choose r}} |A_I|$.
\end{lemma}
\begin{proof}
 Let us define $|A| = \sum_{r=2}^{m} a_r \sum_{I \in {[m] \choose r}} |A_I|$ such that after $r$ terms, all the elements 
 which are present in at most $r$ many sets have been counted exactly once. Hence it is enough to prove that,
 $$a_r = (-1)^r (r-1),\forall 2\le r \le m.$$
 We will prove the above statement by induction on $r$.\\
 \textit{Base case :} $a_2 = 1$ since no elements and particularly the elements present in exactly two sets have not been counted.\\
 \textit{Inductive step :} Let us assume the result for $r \le l-1$.\\
 Now till the first $l-1$ terms, the elements present in exactly $l$ many sets have been counted exactly $\sum_{i=2}^{l-1} a_i {l \choose i}$ 
 many times. Hence,
 $$
 \begin{array}{rclcl}
  a_l &=& 1 - \sum_{i=2}^{l-1} a_i {l \choose i}\\
      &=& 1 - \sum_{i=2}^{l-1} (-1)^i (i-1) {l \choose i}\\
      &=& (-1)^l (l-1)
 \end{array}
 $$
 The second step follows from the inductive hypothesis and the third step follows from Lemma\nobreakspace \ref {sum}. \qed
\end{proof}
\begin{lemma}
 For fixed $n$ and \textit{large} $m$, 
 if in a voting profile $(\succ_1, \dots, \succ_n) \in \mathcal{L(C)}^n$, 
 $\exists c_i,c_j \in \mathcal{C}, c_i \ne c_j$ such that both $c_i$ and $c_j$ are ranked in the top 
 $l$ positions in each preference $\succ_i, 1\le i \le n$ where $(l-1)n<m$, then the profile is manipulable.
\end{lemma}
\begin{proof}
 Since we are working with voting rules and \textit{not} with voting correspondences, the winner is unique. Hence 
 WLOG WMA\footnote{WLOG : Without Loss Of Generality. WMA : We May Assume.} that $c_i$ is not the winner. Let 
 $c_w$ be the winner. Then,
 $$ score(c_w) \le (m-1)n \text{ and } score(c_i) \ge (m-l)n $$
 Hence we have,
 $$ score(c_w) - score(c_i) \le (l-1)n \le m $$
 The above statement along with \MakeUppercase proposition\nobreakspace \ref {manborda} proves the lemma. \qed
\end{proof}

\end{document}